\documentclass[11pt]{article}
\usepackage{amsmath,fullpage}
\usepackage{amsfonts}
\usepackage{amssymb}
\newtheorem{theorem}{Theorem}

\newtheorem{claim}[theorem]{Claim}

\newtheorem{condition}{Condition}
\newtheorem{corollary}[theorem]{Corollary}

\newtheorem{lemma}[theorem]{Lemma}

\newtheorem{proposition}[theorem]{Proposition}

\newenvironment{proof}[1][Proof]{\noindent\textbf{#1.} }{\ \rule{0.5em}{0.5em}}

\begin{document}

\author{Ron Adar\thanks{Department of Computer Science, University of Haifa, Haifa,
Israel. \texttt{radar03@csweb.haifa.ac.il.}}
\and Leah Epstein\thanks{Department of Mathematics, University of Haifa, Haifa,
Israel. \texttt{lea@math.haifa.ac.il.}}}
\title{\textbf{Models for the $\boldsymbol{k}$-metric dimension}}
\date{}
\maketitle

\begin{abstract}
For an undirected graph $G=(V,E)$, a vertex $\tau\in V$ separates
vertices $u$ and $v$ (where $u,v\in V$, $u\neq v$) if their
distances to $\tau$ are not equal. Given an integer parameter $k
\geq1$, a set of vertices $L\subseteq V$ is a feasible solution if
for every pair of distinct vertices, $u,v$, there are at least $k$
distinct vertices $\tau_{1},\tau_{2},\ldots,\tau_{k}\in L$ each
separating $u$ and $v$. Such a feasible solution is called a
\textit{landmark set}, and the $k$-metric dimension of a graph is
the minimal cardinality of a landmark set for the parameter $k$.
The case $k=1$ is a classic problem, where in its weighted
version, each vertex $v$ has a non-negative weight, and the goal
is to find a landmark set with minimal total weight. We generalize
the problem for $k \geq2$, introducing two models, and we seek for
solutions to both the weighted version and the unweighted version
of this more general problem. In the model of all-pairs (AP), $k$
separations are needed for every pair of distinct vertices of $V$,
while in the non-landmarks model (NL), such separations are
required only for pairs of distinct vertices in $V \setminus L$.

We study the weighted and unweighted versions for both models (AP
and NL), for path graphs, complete graphs, complete bipartite
graphs, and complete wheel graphs, for all values of $k \geq2$. We
present algorithms for these cases, thus demonstrating the
difference between the two new models, and the differences between
the cases $k=1$ and $k \geq2$.

\end{abstract}

\section{Introduction}

The problem of finding a landmark set or a resolving set of a
graph was studied in a number of papers
\cite{Slat75,HM1976,Babai,Chvatal,MT,KRR1996,CEJO00,SSH02,ST04,HSV11,ELW,HN12}.
In this problem, one is interested in finding a subset of vertices
$L$ of an undirected graph $G=(V,E)$, such that the ordered list
of distances of a vertex $u$ to the vertices of $L$ uniquely
determine its identity. That is, the set $L\subseteq V$ should be
such that for any $u,v\in V$ ($u\neq v$), there exists $\tau\in L$
such that $d(u,\tau)\neq d(v,\tau)$ (letting $d(x,y)$ denote the
number of edges in a shortest path between $x$ and $y$). We say
that such a vertex $\tau$ separates $u$ and $v$, and a set $L$
satisfying the property that for every pair of distinct graph
vertices it contains a vertex that separates them, is called a
landmark set. Since for any $\ell\in L$, $d(\ell,\ell)=0$ and
$d(\ell,x)>0$ for any $x\neq\ell$, every landmark $\ell$ separates
itself from any vertex $v\neq\ell$, and the requirement that a
separating vertex $\tau$ will exist for any $u,v\in V\setminus L$
($u\neq v$) is equivalent to the requirement that $\tau$ will
exist for any distinct pair of vertices of $V$. For a given graph,
the algorithmic problem of finding a landmark set of minimum
cardinality is called the metric dimension problem (and this
minimum cardinality is the metric dimension). In the weighted
case, a non-negative rational weight is given for each vertex by a
function $w:V\rightarrow \mathbb{Q}^{+}$. For a set $U\subseteq
V$, $w(U)={\sum_{v\in U}}w(v)$ is the total weight of vertices in
$U$, and the goal is to find a landmark set $L$ with the minimum
value $w(L)$, where its weight is called the weighted metric
dimension.

One limitation of this model in large networks is the frequent
failures of vertices. When a failure occurs at a landmark vertex,
this can harm the identification process for other vertices. We
generalize the metric dimension problem and require that $L$ will
contain $k$ or more separating vertices for every pair of distinct
vertices for a given parameter $k\geq2$. This problem is called
the $k$-metric dimension problem, and its weighted variant is
called the weighted $k$-metric dimension problem. There can be two
types of requirements for landmark sets. The first one is to
separate (by $k$ or more separations) any pair of distinct
vertices of $V$, and the second one considers separations for any
pair of distinct vertices of $V\setminus L$. We consider both
these models in this work. We refer to the first variant as the
all-pairs model (AP), and the second variant is called the
non-landmarks model (NL). We will demonstrate the differences
between landmark sets for the case $k=1$ and for larger values of
$k$, and the differences between solutions for the two models.
This is done using several graph classes with simple structures.
These graph classes are path graphs, complete graphs (cliques),
complete bipartite graphs, and complete wheels\footnote{In an
accompanying paper \cite{AEsec}, we analyze additional graph
classes with respect to the concepts that are introduced in this
paper.}.

Given the input graph $G=(V,E)$, let $n=\left\vert V\right\vert $
be the number of its vertices. The set of landmark sets of $G$ for
the parameter $k$ ($k\geq2$) is denoted by $LS_{k}^{M}(G)$, where
$M=AP$ or $M=NL$ (we use $M$ when we discuss a specific model and
not some general facts that are not model dependent). The minimum
cardinality of any landmark set, which is the cost of an optimal
solution for the unweighted problem, is denoted by
$md_{k}^{M}(G)$. A set $Q\in LS_{k}^{M}(G)$, such that $w(Q)\leq
w(D)$ for any set $D\in LS_{k}^{M}(G)$, is a solution for the
weighted problem, and $wmd_{k}^{M}(G)=w(Q)$ denotes the weight of
this solution. For a vertex $v\in V$ we use the notation
$N_{v}(G)=\{u\in V:d(u,v)=1\}$, that is, all neighbors of $v$. In
addition, for $u,v\in V$ we use the notation
$N_{u,v}(G)=N_{u}(G)\cap N_{v}(G)$ for all the common neighbors of
$u$ and $v$ in $G$ and the notation $SpS_{u,v}(L)=\{z\in
L:d(u,z)\neq d(v,z)\}$ which denotes the set of all vertices in
$L\subseteq V$ that separate $u$ and $v$. If there is no feasible
solution, we will write $md_{k}^{M}=\infty$ and
$wmd_{k}^{M}=\infty$. This case is possible for AP, but for NL, as
we will explain below, a feasible solution always exists.

By definition, for any connected undirected graph $G$, if $L\in
LS_{k}^{AP}(G)$, then $L\in LS_{k}^{NL}(G)$. Thus, for any integer
$k\geq2$, $md_{k}^{NL}(G)\leq md_{k}^{AP}(G)$, and
$wmd_{k}^{NL}(G)\leq wmd_{k}^{AP}(G)$ for any non-negative weight
function $w:V\rightarrow\mathbb{Q}^{+}$ on $G$'s vertices. Another
obvious property is $LS_{k+1}^{M}(G) \subseteq LS_{k}^{M}(G)$ for
any $k\geq 1$ and any model $M$. The next lemma will be used for
analyzing landmark sets.

\begin{lemma}
\label{k=2 lemma} Let $L \subseteq V$, and let $u,v\in L$ be a distinct pair
of vertices. Then $u$ and $v$ are separated by at least two vertices of $L$.
\end{lemma}
\begin{proof}
Since $u,v\in L$, $d(u,u)=d(v,v)=0$ and $d(u,v)\neq0$ since $u\neq v$, thus
$u$ and $v$ separate each other.
\end{proof}

In the body of the paper we assume $n\geq2$. We now discuss the
cases $n=1,2$, and several other simple cases where the solutions
do not depend on the structure of the graph. The case $n=1$ is
trivial, as in this case $\emptyset$ is a valid landmark set for
both models and any $k\geq1$. For AP, if $2\leq n\leq k-1$, there
are no feasible solutions, as a feasible solution must have at
least $k$ landmarks. In the case $n=k=2$, $V$ is the unique
landmark set as a landmark set for AP must have at least $k$
vertices, and by Lemma \ref{k=2 lemma}. We will show in Section
\ref{pat} that in the case $n=k\geq3$, there is no feasible
solution for AP. On the other hand, for NL, any subset of $n-1$
vertices is a feasible solution for any $k$, and thus a feasible
solution always exists. We call such a solution \textit{a trivial
solution}. A solution consisting of $n$ vertices is also feasible,
but we will never use such a solution for NL as a solution of
$n-1$ vertices always has smaller cardinality and smaller weight.
Any solution that is not trivial must consist of at least $k$
vertices. Thus, for $2\leq n\leq k+1$, any optimal solution is a
trivial solution. On the other hand, if $n\geq k+2$, a solution
consisting of at most $k-1$ vertices cannot be valid. Thus, in
particular, if $n=2$, then $md_{2}^{AP}=2$ and $md_{2}^{NL}=1$,
and for $k\geq3$ (and $n=2$), $md_{k}^{AP}=\infty$ and
$md_{k}^{NL}=1$. For any graph $G$, $V\in LS_{2}^{AP}(G)$, so a
landmark set always exists in this case. However, there are graphs
where there is no feasible solution even if $n>k\geq3$ (see
Section \ref{clik}). Given the cases that are completely resolved
independently of the graph structure, in what follows, we will
deal with $n\geq3$ and $n\geq k+1$ for AP (and we will prove that
there is no solution for $n=k\geq3$), and with $n\geq k+2$ for NL.

As in \cite{ELW}, we are often interested in minimal landmark sets
with respect to set inclusion, since any minimum weight landmark
set must be such a set. Note that in many cases not every minimal
landmark set (with respect to set inclusion) is a minimum
cardinality landmark set. An algorithm for finding a minimum
weight landmark set can test minimal landmark sets to find one of
minimum weight (by enumerating them or via dynamic programming).

\paragraph{Previous work.}

As mentioned above, most previous studies dealt with the case
$k=1$. In \cite{CW14}, a variant where separation is defined using
pairs  of vertices was studied (see also \cite{CH+07}). That model
is different from ours, and in particular, the condition for a
pair of landmark vertices $u,v$ to doubly separate a pair of graph
vertices $x,y$ is defined on all four vertices together, and it is
required that $d(u,x)-d(u,y)\neq d(v,x)-d(v,y)$. Note that by
definition, $u,v$ cannot satisfy this if $u=v$, but $u,v$ doubly
separate $u,v$. A double landmark set $S$ is a subset of $V$ such
that for any $x,y\in V$, there exist $u,v\in S$ that doubly
separate $x$ and $y$. There is no particular relation between this
condition and the conditions that we require for $k\geq2$ (neither
for NL nor for AP), as can be seen in the following examples.
Consider a path graph on $n\geq5$ vertices;
$v_{1},v_{2},\ldots,v_{n}$ (in this order). In Section \ref{pat},
it is shown that the subset $X=\{v_{2},v_{3},\ldots,v_{n-1}\}$ is
a landmark set for $k \leq n-3$ (for both AP and NL). This is,
however, not a double landmark set, as the
$d(u,v_{1})-d(v,v_{1})=d(u,v_{2})-d(v,v_{2})$ for any $u,v \in X$
(see \cite{CW14}). On the other hand, if the graph is a wheel with
seven vertices $c_1,\ldots,c_7$ (where $c_7$ is the central
vertex, see Section \ref{whel}), by the results of \cite{CW14},
$c_{1},c_{3},c_5$ is a double landmark set, while it gives only
one separation for $c_{7}$ and $c_{2}$, and thus it is not a
landmark set even for $k=2$ (for both AP and NL). This difference
holds for larger wheels as well, as can be seen from the fact that
the conditions that we present for landmark sets of wheels in
Section \ref{whel} differ from the condition given in \cite{CW14}
for double landmark sets.

This metric dimension problem was introduced over forty years ago
by Harary and Melter \cite{HM1976} and by Slater \cite{Slat75}. It
was studied widely in the combinatorics literature, where exact
values of the metric dimension or bounds on it for specific graph
classes are obtained \cite{Babai,CH+07,CEJO00,CZ03}. The problem
was also studied with respect to complexity in the sense that it
was shown to be NP-hard in general \cite{KRR1996}, it was shown
that it is hard to approximate \cite{BE+06,HSV11,HN12}, and it is
even NP-hard for certain graph classes \cite{DPL11,ELW}. On the
other hand, it is polynomially solvable (often even for the
weighted version) for many graph classes; in particular such
classes include paths, cycles, trees, and wheels
\cite{Slat75,HM1976,KRR1996,CEJO00,SSH02}. The main graph classes
studied here are paths and wheels, which were studied in
\cite{KRR1996,ELW} and \cite{HM1976,SSH02,ELW}, respectively. It
was proved by Khuller, Raghavachari, and Rosenfeld \cite{KRR1996}
that the metric dimension (for $k=1$ and the unweighted case) of a
path is $1$, and was shown by Shanmukha, Sooryanarayana and
Harinath \cite{SSH02} that the metric dimension is
$\lceil\frac{2n}{5}\rceil$ for a complete wheel with $n\geq8$
vertices (for $n=4,5,6$, and $7$ the metric dimension is equal to
$3,2,2$ and $3$, respectively). In the weighted case, it is shown
in \cite{ELW} that a minimum weight landmark set for a path
consists of one or two vertices. For complete wheels, a condition
on the positions of landmarks on the cycle of the wheel was stated
\cite{ELW} (see Section \ref{whel}). Here, we define several
conditions for the different cases (according to the different
values of $k$, and the two models), where these conditions differ
from that of the case $k=1$ in all cases. We also study complete
graphs, for which any landmark set for $k=1$ consists of at least
$n-1$ vertices \cite{CEJO00}, and complete bipartite graphs, where
(if the graph is connected and $n\geq3$), any landmark set
contains all vertices except for at most one vertex of each
partition \cite{CEJO00} (for $k=1$).

\section{Paths}
\label{pat}

We start with path graphs. These graphs are typically easier to
analyze than any other graph class. We will show that while in AP
landmark sets for $k\geq 2$ have similar structures to those of
the case $k=1$, in NL this is not the case.

\subsection{All-pairs model (AP)}

In this section, we let $G=(V,E)$ be a path graph with $n\geq3$. Let $k\leq
n-1$, and let $V=\{v_{1},\ldots v_{n}\}$, where the vertices appear on the
path in this order, i.e., $E$ consists of the edges $\{v_{i},v_{i+1}\}$ for
$1\leq i\leq n-1$. The vertices $v_{1},v_{n}$ are called end vertices, while
the other vertices are called internal vertices (and every path where $n\geq3$
has at least one such vertex). We will adapt the following property.

\begin{claim}
\label{rem k=1} For $k=1$, $md_{1}^{AP}(G)=1$, and a minimum cardinality
landmark set consists of either $v_{1}$ or $v_{n}$ \cite{KRR1996}. Any minimal
(with respect to set inclusion) landmark set that is not $\{v_{1}\}$ or
$\{v_{n}\}$ consists of exactly two internal vertices \cite{ELW}.
\end{claim}

\begin{proposition}
\label{line prop} For $n\geq2$, $md_{2}^{AP}(G)=2$, and
$md_{k}^{AP}(G)=k+1$ for $k\geq3$. Moreover, for $k\geq3$, a set
$L\subseteq V$ is a minimal landmark set if and only if $|L|=k+1$.
For $k=2$, a set $L\subseteq V$ is a minimal landmark set (with
respect to set inclusion) if and only if either
$L=\{v_{1},v_{n}\}$ or $|L|=3$ and $\{v_{1},v_{n}\}\subsetneq L$.
A minimum cardinality or minimum weight landmark set can be found
in time $O(n)$.
\end{proposition}
\begin{proof}
\label{path proof1} We start with showing that any subset $X$ of $k+1$
vertices is a landmark set. This holds as for any pair $v_{i},v_{j}\in V$
($i\neq j$), there is at most one vertex of equal distances to both of them.
Specifically, $v_{(i+j)/2}$ has equal distances to $v_{i}$ and $v_{j}$, if
$i+j$ is even (and otherwise, if $i+j$ is odd, then there is no such vertex).
Thus, there are $k$ or $k+1$ separations vertices in $X$ for any pair
$v_{i},v_{j}$.

Recall that a landmark set for AP must contain at least $k$ vertices. We show
that a set $X$ consisting of $k$ vertices such that at least one of them is
internal is not a landmark set. Assume that $v_{a}\in X$, where $1<a<n$. The
two vertices $v_{a-1}$ and $v_{a+1}$ have equal distances to $v_{a}$, and
therefore they have at most $k-1$ separations. This shows that except for the
case $k=2$, any landmark set has exactly $k+1$ vertices. Since any subset of
$k+1$ vertices is a landmark set for $k\geq2$, we find that for $k\geq3$, the
class of minimal landmark sets is exactly the class of subsets of $k+1$
vertices. Finding a minimum weight such subset can be done in linear time by
selecting $k+1$ vertices of minimum weights (breaking ties arbitrarily).

For $k=2$, $\{v_{1},v_{n}\}$ is in fact a landmark set, as the
distances of graph vertices from $v_{1}$ and from $v_{n}$ are
unique, and it is the only landmark set of minimum cardinality.
Thus, any set of the form $\{v_{1},v_{i},v_{n}\}$ for $1<i<n$ is
not a minimal landmark set (with respect to set inclusion), while
other sets consisting of three vertices are minimal landmark sets.
Finding a minimum weight landmark set can be done as follows.
Select the minimum weight solution out of the following three
solutions: a minimum weight subset of three internal vertices, a
minimum weight subset of two internal vertices plus a minimum
weight vertex out of the end vertices, and finally, the third
option is the solution $\{v_{1},v_{n}\}$. These solutions can be
computed in linear time as well.
\end{proof}

We finish this section with showing that for any graph with $n=k\geq3$ there
does not exist a landmark set .

\begin{claim}
For any graph with $n=k\geq3$, there is no feasible solution for AP.
\end{claim}
\begin{proof}
Assume by contradiction that there is a landmark set $L\subseteq
V$ for a graph $G$ with $n=k\geq3$. Since $n=k$, in order to
obtain $k$ separations for any pair of distinct vertices, we find
$L=V$, and every vertex of $V$ must separate any pair of distinct
vertices of $V$. As distances in the graph can take values in
$\{0,1,\ldots,n-1,\infty\}$, the graph must be connected. This
holds as in a disconnected graph, there exists a pair of vertices
$u,v\in V$ where $d(u,v)=\infty$, and in this case for any vertex
$x\in V$, $x\neq u,v$ ($x$ must exist as $n\geq3$), at least one
of $d(x,u)=\infty$ and $d(x,v)=\infty$ must hold since otherwise
$d(u,v)<\infty$, and either $u$ does not separate $x$ and $v$ or
$v$ does not separate $x$ and $u$ (or both).

Since the graph is connected, given a vertex $v \in V$, there are
$n-1$ distinct values for the distances $d(v,y)$ for $y\in V$, and
there exists $y'$ such that $d(v,y')=n-1$. The graph must be a
path connecting $v$ and $y'$. Let $x'\in V\setminus\{v,y'\}$
(where $x'$ must exist as $n \geq 3$). Similar reasoning shows
that there is a vertex $y''$ ($y'' \in V \setminus\{x'\}$ such
that the graph is a path connecting $x'$ and $y''$. As $x\neq
v,y'$, we reach a contradiction.
\end{proof}

\subsection{Non-landmarks model (NL)}

In this section we will show that $md_{k}^{NL}(G)=k$ for any $k\geq2$ and
$n\geq k+2$. Since any subset of $k+1$ vertices is a landmark set for AP and
therefore also for NL, we focus on subsets of $k$ vertices. Any such landmark
set is obviously minimal with respect to set inclusion. Given a set
$L\subseteq V$, we say that a hole is a vertex of $V\setminus L$. A maximum
length sequence of consecutive holes is a called gap (induced by $L$). The
length of a gap is defined to be the number of vertices is it. A gap can
possibly contain an end vertex. As $n>k$, if $|L|=k$, there is at least one
hole, thus at least one gap induced by $L$. Similarly, a vertex $v_{i}\in L$
is called an anti-hole, and a sequence of consecutive vertices of $L$ is
called an anti-gap. \medskip

The following two lemmas deal with the case where $L\subseteq V$
is a set of vertices that satisfies $|L|=k$.

\begin{lemma}
If $L$ induces exactly one gap, then $L$ is a landmark set. Thus,
$md_{k}^{NL}=k$.
\end{lemma}
\begin{proof}
Since there is just one gap, $L=\{v_{1},\ldots,v_{a},v_{b},\ldots,v_{n}\}$,
such that $0 \leq a < b \leq n+1$ and $a+(n-b+1)=k$ (if $a=0$, then $v_{1}
\notin L$, and if $b=n+1$, then $v_{n} \notin L$). Consider the vertices of
$V\setminus L$, that is, $v_{a+1},\ldots,v_{b-1}$. For any pair of vertices
$v_{i},v_{j}$, such that $a<i<j<b$, the distances to any vertex in $L$ are
distinct (as the path from $v_{j}$ to any vertex $v_{a^{\prime}}$ with
$a^{\prime}\leq a$ traverses $v_{i}$, and the path from $v_{i}$ to any vertex
$v_{b^{\prime}}$ with $b^{\prime}\geq b$ traverses $v_{j}$). As such a subset
$L$ must exist (for example, $\{v_{1},v_{2},\ldots,v_{k}\}$ is such a set), we
have $md_{k}^{NL}=k$.
\end{proof}

\begin{lemma}
If $L$ induces at least two gaps, and at least one gap has at least two
vertices, then $L$ is not a landmark set.
\end{lemma}

\begin{proof}
Assume that $L$ induces two such gaps. Consider two consecutive
gaps (where the set of vertices between them is an anti-gap) such
that at least one of these gaps has two vertices. Let $v_{i}$ be
the last vertex of the first gap out of the two, and let $v_{j}$
($j\geq i+2$) be the first vertex of the second gap. Without loss
of generality assume that $v_{j+1}$ is a hole as well (the proof
for the case that $v_{i-1}$ is a hole is similar). If $i+j$ is
even, then the vertex $v_{(i+j)/2}$ has equal distances to $v_{i}$
and $v_{j}$ and does not separate them, and $v_{(i+j)/2}\in L$ as
$i<(i+j)/2<j$, and all vertices on the path between $v_{i}$ and
$v_{j}$ are in $L$ (as $v_{i+1},\ldots,v_{j-1}$ is an anti-gap).
As $|L|=k$, there are at most $k-1$ separations between $v_{i}$
and $v_{j}$, a contradiction. If $i+j$ is odd, then the vertex
$v_{(i+j+1)/2}$ satisfies $i<(i+j+1)/2<j$, so $v_{(i+j+1)/2}\in
L$, and it has equal distances to $v_{i}$ and to $v_{j+1}$, and
there are at most $k-1$ separations between these two holes, a
contradiction again.
\end{proof} \medskip

In the next three lemmas, $L\subseteq V$ is a landmark set that
satisfies $|L|=k$.

\begin{lemma}
\label{even anti-gap lemma}If $L$ induces at least two gaps, then the length
of every anti-gap that does not contain an end vertex is even.
\end{lemma}

\begin{proof}
Assume by contradiction that there is an anti-gap of odd length
$v_{i},\ldots,v_{i+2x}$ for some integer $x\geq1$, where $i>1$ and
$i+2x<n$, such that $v_{i-1}$ and $v_{i+2x+1}$ are holes. The
vertex $v_{i+x}$ is a part of the anti-gap as $i<i+x<i+2x$, but it
does not separate the holes $v_{i-1}$ and $v_{i+2x+1}$, as it has
equal distances of $x+1$ to both of them. Thus, there are at most
$k-1$ separations between these anti-holes, a contradiction.
\end{proof}

\begin{lemma}
\label{anti-gap length lemma}If $L$ induces at least three gaps, then all
anti-gaps that do not contain end vertices have the same length.
\end{lemma}

\begin{proof}
Assume by contradiction that there are two anti-gaps of different lengths that
do not have end vertices. Thus, there are two such consecutive anti-gaps (with
even numbers of vertices) where the first sequence is of the form
$v_{i},\ldots,v_{i+2x-1}$, and the second sequence is of the form
$v_{i+2x+1},\ldots,v_{i+2x+2y}$, such that $v_{i-1}$, $v_{i+2x}$, and
$v_{i+2x+2y+1}$ are holes (where $x,y$ are integers such that $x,y\geq1$ and
$x\neq y$, as the anti-gaps have lengths $2x$ and $2y$, that are distinct even
integers). We claim that the vertex $v_{i+x+y}$ is an anti-hole. The only hole
among $v_{i},\ldots,v_{i+2x+2y}$ is $v_{i+2x}$ (obviously $v_{i+x+y}$ is one
of the vertices $v_{i},\ldots,v_{i+2x+2y}$ as $i<i+x+y<i+2x+2y$). The vertices
$v_{i+2x}$ and $v_{i+x+y}$ are distinct as $x\neq y$. The distances of
$v_{i+2x+2y+1}$ and $v_{i-1}$ from $v_{i+x+y}$ are both equal to $x+y+1$, and
therefore there are at most $k-1$ separations between $v_{i+2x+2y+1}$ and
$v_{i-1}$ in $L$, a contradiction.
\end{proof}

\begin{lemma}
If $L$ induces at least two gaps, then the set of the $n-k$ holes
is of the form
$$v_{i},v_{i+\rho},v_{i+2\rho},\ldots,v_{i+(n-k-1)\rho}$$ where
$\rho \geq3$ is odd, $i\geq1$, and $i+(n-k-1)\rho\leq n$. If $G$
has a landmark set with $k$ vertices that induces at least two
gaps, then $n\leq3k/2+1$.
\end{lemma}

\begin{proof}
As there are $k$ anti-holes, there are $n-k$ holes. Every gap consists of a
single hole. Let $i$ be the hole of smallest index. Let $v_{i+\rho}$ be the
next hole. The number $\rho-1$ must be a positive even number (by Lemma
\ref{even anti-gap lemma}), and thus $\rho\geq3$, and $\rho$ is odd. All other
anti-gaps excluding the ones that possibly occur before the first hole and
after the last hole also have lengths of $\rho-1$ (by Lemma
\ref{anti-gap length lemma}) and gaps consisting of a single hole separate them.

Next, we show the second claim. We have $i+(n-k-1)\rho\leq n$ while
$n-k-1\geq1$, $\rho\geq3$ and $i\geq1$, implying $1+3(n-k-1)\leq n$, or
alternatively, $n\leq3k/2+1$.
\end{proof}

We now state the algorithm for all values of $k$ and $n\geq k+2$.
An algorithm for finding a minimum cardinality landmark set simply
returns vertices $v_{1},\ldots,v_{k}$. In the weighted variant,
three solutions are calculated, and the output is a minimum weight
solution out of these three solutions. The first solution is a set
of $k+1$ minimum weight vertices. This solution is computed in
time $O(n)$. The second solution is a minimum weight solution out
of those with a single gap. These are solutions of the form
$\{v_{1},\ldots,v_{a},v_{b},\ldots,v_{n}\}$ for $0\leq a<b\leq
n+1$, where $a+(n-b+1)=k$, and these solutions can be enumerated
in time $O(n)$. If $n\leq3k/2+1$, solutions inducing at least two
gaps are considered too. For any odd integer
$3\leq\rho\leq\frac{n-1}{k-1}$, we find the maximum weight of a
sequence of the form
$v_{i},v_{i+\rho},v_{i+2\rho},\ldots,v_{i+(n-k-1)\rho}$ by testing
all relevant values of $i$. This can be done in time $O(n)$ for
each value of $\rho$, and in time $O(n^{2})$ in total.

\section{Complete wheel graphs}
\label{whel} This section deals with the case where $G$ is
complete wheel graph, where $V= C\cup\{h\}$ such that
$C=\{c_{1},...,c_{n-1}\}$ is $G$'s cycle, $h$ is the central
vertex, and the other vertices are cycle vertices. For simplicity,
we extend the definition of indices of cycle vertices as follows.
The vertex $c_{i}$ for $i\geq n$ or $i\leq0$ is defined as
$c_{i}=c_{b}$, for an integer $b$ where $1 \leq b \leq n-1$ such
that $b=b+j(n-1)$ for an integer $j$ (where $j$ may be positive or
negative). In such a graph, $E=\{\{c_{i},h\}|1 \leq i \leq
n-1\}\cup\{\{c_{i},c_{i+1}\}|1 \leq i \leq n-1\}$.

In wheel graphs, any distance between any possible pair of vertices is either
$1$ or $2$. Distances of $1$ occur only between pairs of vertices connected by
an edge, that is, pairs of neighboring cycle vertices, and pairs consisting of
$h$ and a cycle vertex. Note that $n\geq4$, since $G$'s cycle has at least $3$
vertices. Note, also, that for a pair of vertices $\{u,v\}\subset V$,
$N_{u,v}(G)$ is never empty as it contains $h$ if $u,v\neq h$, and if
$\{u,v\}=\{h,c_{i}\}$, then $N_{u,v}(C)=\{c_{i-1},c_{i+1}\}$ (however, in the
case that we only consider neighbors on $G$'s cycle, $C$, it is possible that
$N_{u,v}(C)=\emptyset$). We start the discussion by examining small wheels;
these cases are slightly different from larger wheels as there is no fixed
pattern with respect to possible positions of landmarks. Afterwards, we
discuss the more general case of larger wheels, where we will show that for
$n\geq9$, the central vertex can be omitted for any landmark set $L\subseteq
V$, for both models (i.e. if $L\in LS_{k}^{M}(G)$ then $L\backslash\{h\}\in
LS_{k}^{M}(G)$). Since for $n=4$, a complete wheel graph is actually a
complete graph (or a clique graph), we deal with the more general case of
complete graphs first.

\subsection{Complete graphs (cliques) and complete bipartite graphs}

\label{clik} Let $G=(V,E)$ be a complete graph with $n \geq3$ (i.e.
$E=\{\{u,v\}|u,v,\in V\}$). In this graph the distances satisfy $d(u,v)=1$ for
any $u,v\in V$.

\begin{proposition}
\label{clique} For a complete graph $G$, $md_{2}^{AP}(G)=n$, and
$md_{k}^{AP}(G)=\infty$ for $k>2$. Additionally,
$md_{k}^{NL}(G)=n-1$. For any $G$ and $k$, a minimum weight
landmark set can be computed in linear time.
\end{proposition}
\begin{proof}
For a pair of vertices $u\neq v$, any third vertex $x\neq u,v$ does not
separate them. Thus, a landmark set for any model must contain at least $n-1$
vertices. Since any subset of $n-1$ vertices is a trivial landmark set for NL,
we find $md_{k}^{NL}(G)=n-1$. Next, consider AP. Since for any pair of
vertices $u,v$, $SpS_{u,v}(L)=\{u,v\}\cap L$, at most two separations are
possible, so $md_{k}^{AP}(G)=\infty$ for $k>2$, and for $k=2$, $L$ must
contain all vertices to allow two separations for every pair. We briefly
discuss algorithms for minimum weight landmark sets. For AP, the algorithm
returns $V$ if $k=2$, and otherwise it reports that there is no feasible
solution. For NL, it finds a vertex $y$ of maximum weight and returns
$V\setminus\{y\}$.
\end{proof}

Next, let $G=(V,E)$ be a complete bipartite graph with $n\geq3$ with the
partitions $A,B$ (where $A,B\neq\emptyset$, $A\cup B=V$, $A\cap B=\emptyset$
and $E=\{\{u,v\}|u\in A,v\in B\}$). In this graph the distances satisfy
$d(u,v)=1$ for $u\in A$ and $v\in B$, or if $u\in B$ and $v\in A$ and
$d(u,v)=2$ if $u\neq v$, $u,v\in A$ or $u,v\in B$.

\begin{proposition}
\label{compbip} For a complete bipartite graph $G$, $md_{2}^{AP}(G)=n$, and
$md_{k}^{AP}(G)=\infty$ for $k>2$. Additionally, $md_{k}^{NL}(G)=n-2$ for
$n\geq k+2$. For any $G$ and $k$, a minimum weight landmark set can be
computed in linear time.
\end{proposition}
\begin{proof}
For a pair of vertices $u\neq v$ that belong to the same partition ($A$ or
$B$), any third vertex $x\neq u,v$ does not separate them. Thus, a landmark
set for AP must contain all vertices proving the claim for AP. By the same
reasoning, for NL, a landmark set must contain all vertices of $A$ except for
at most one vertex and all vertices of $B$ except for at most one vertex.
Given $a\in A$ and $b\in B$, $V\setminus\{a,b\}$ is a landmark set since
$k\leq n-2$, every vertex of $B$ (excluding $b$) has distance $2$ to $b$ and
distance $1$ to $a$, and every vertex of $A$ (excluding $a$) has distance $1$
to $b$ and distance $2$ to $a$. For NL, an algorithm finds a pair of vertices
(one from $A$ and one from $B$ of maximum weights and returns $V\setminus
\{a,b\}$.
\end{proof}

\subsection{General properties and the case {${\boldsymbol{n=5}}$}}

After covering the case of a complete graph, we get back to the case $G$ is a
complete wheel graph as described above, where $n \geq5$. We continue with
several simple properties that will be used throughout the section.

\begin{claim}
\label{separatingclaim} Let $u,v\in G$ be cycle vertices. Then $u$
and $v$ can be separated by vertices in $((N_{u}(G)\cup
N_{v}(G))\setminus N_{u,v}(G))\cup\{u,v\}$. Moreover, $u$ and $h$
can be separated by vertices in $(C \setminus
N_{u}(G))\cup\{u,h\}$.
\end{claim}

\begin{proof}
As all distances in the graph are in $\{0,1,2\}$, a vertex $x \neq u,v$ that
separates $u$ and $v$ has distance $1$ to one of these vertices and distance
$2$ to the other. It is therefore a neighbor of one of them but not of the
other, proving the first part. As the distance of $h$ to any cycle vertex is
$1$, the only two cycle vertices that cannot separate $h$ and $u$ are the two
neighbors of $u$.
\end{proof}

We say that $c_{i}$ and $c_{j}$ are distant if both paths on the cycle between
$c_{i}$ and $c_{j}$ are of length at least $3$, that is, each such path has at
least two cycle vertices between them. In this case $j=i+\ell$ where
$3\leq\ell\leq n-4$. If $n\in\{5,6\}$, then no distant pairs of vertices exist.

\begin{corollary}
Let $1\leq i\leq n-1$, and consider a set $L\subseteq V$. If
$n\geq5$, we have
$SpS(c_{i},c_{i+1})(L)=L\cap\{c_{i-1},c_{i},c_{i+1},c_{i+2}\}$, if
$n\geq6$,
$SpS(c_{i},c_{i+2})(L)=L\cap\{c_{i-1},c_{i},c_{i+2},c_{i+3}\}$,
and if $n\geq7$, for any pair of distant vertices $c_{i},c_{j}$,
$SpS(c_{i},c_{j})(L)=L\cap\{c_{i-1},c_{i},c_{i+1},c_{j-1},c_{j},c_{j+1}\}$.
\end{corollary}

\begin{proof}
We use Claim \ref{separatingclaim} for the three cases.
\end{proof}

Next, we provide an analysis of the case $n=5$. It is obvious that a minimum
cardinality or minimum weight landmark set can be found in constant time by
enumerating all subsets of vertices for constant values of $n$. In this
section we will find all minimal landmark sets (with respect to set
inclusion), which allows an algorithm to enumerate a smaller set of candidate
subsets. In the next section this is done for $6\leq n\leq8$ as well.

\begin{theorem}
Let $n=5$.

\begin{itemize}
\item For AP, $md_{2}^{AP}(G)=4$, $wmd_{2}^{AP}(G)=w(C)$, and
$md_{3}^{AP}(G)=\infty$.

\item For NL, $md_{2}^{NL}(G)=3$, and
\[
wmd_{2}^{NL}(G)=\min\{w(C),w(V)-\max_{1 \leq i \leq4} (w(c_{i})+w(c_{i+1}))\}
\ .
\]
Additionally, $md_{k}^{NL}(G)=4$ for $k\geq3$, and
\[
wmd_{k}^{NL}(G)=\min\{w(C), W(V)-\max_{1 \leq i \leq4} w(c_{i}) \} \ .
\]

\end{itemize}
\end{theorem}

\begin{proof}
First, consider AP and $k=2$. A pair of cycle vertices $c_{i}$ and $c_{i+2}$
cannot be separated by any vertex except for $c_{i}$ and $c_{i+2}$, as they
have distances of $1$ to all other vertices. Thus, all the cycle vertices
(which can be split into two such pairs) must belong to any landmark set. We
claim that $C$ is indeed a landmark set (and it is a unique minimal landmark
set). By Lemma \ref{k=2 lemma}, it is sufficient to show two separations
between $c_{i}$ (for $1\leq i\leq4$) and $h$. One separation results from
$c_{i}$ being a member of the landmark set. Additionally, $d(h,c_{i+2})=1$
while $d(c_{i},c_{i+2})=2$, so $c_{i+2}$ separates $c_{i}$ and $h$. Note that
$V$ is a landmark set as well, but it is not minimal.

Next, consider NL and $k=2$, and let $L$ be a landmark set. In this case, for
$i=1,2$, at least one of $c_{i}$ and $c_{i+2}$ must belong to $L$ (otherwise
they cannot be separated, as explained above). If there is at least one index
$j$ such that $c_{j}\notin L$ and in addition, $h\notin L$, then as the
distances of $h$ and $c_{j}$ are $1$ both to $c_{j-1}$ and to $c_{j+1}$, only
$c_{j+2}$ can separate them (if it is in $L$), and $L$ cannot be a landmark
set. Thus, if $h\notin L$, then $L=C$, where $C$ is a trivial landmark set. On
the other hand, a set of the form $\{c_{i},c_{i+1},h\}$ is a landmark set, as
$c_{i-1}$ and $c_{i+2}$ are separated by both $c_{i}$ and $c_{i+1}$.
Therefore, a minimal landmark set (with respect to set inclusion) consists of
either of $C$, or $h$ and two adjacent cycle vertices.

Finally, we deal with $k \geq3$. For AP, due to the above, there
cannot be more than two separations between $c_{i}$ and $c_{i+2}$,
so $md_{k}^{AP}=\infty$. For NL and $k \geq4$, a non-trivial
landmark set must contain at least four vertices (otherwise there
cannot be four separations), thus only trivial solutions exist.
For NL and $k=3$, as already for $k=2$ any landmark set with at
most three vertices must contain $h$ and two cycle vertices, for
such a set, the other two cycle vertices will only have two
separations. Thus, for $k=3$ only trivial solutions exist as well.
\end{proof}

\medskip
In what follows, we deal with the case $n \geq6$. The crucial part
of designing an algorithm for each option (where an option
consists of the model and the value of $k$) is to find a suitable
condition on the positions of landmarks on the cycle. In the case
$k=1$, given a landmark set $L\subseteq V$, let a gap be a maximum
cardinality set of consecutive cycle vertices in $V \setminus L$.
Two gaps are consecutive if they are separated by exactly one
landmark (the set of vertices of the two gaps and the single
landmark are consecutive vertices). The condition (defined for $n
\geq 8$ \cite{ELW}) is as follows. There is no gap of four
vertices or more, there is no pair of consecutive gaps such that
both contain more than one vertex, and finally, there is at most
one gap of three vertices. Here, the exact conditions for the two
models are different from the condition for $k=1$, and in fact we
define four conditions for four cases. However, landmark sets for
sufficiently large $n$ ($n \geq 9$) can be computed using one
general approach for all cases (via dynamic programming).

\subsection{All-pairs model (AP)}
\subsubsection{The case $\boldsymbol{k=2}$}

Consider the following condition for $n\geq6$.

\begin{condition}
\label{APk2} For any pair of consecutive vertices of $C$ that are
not in $L$, the three consecutive vertices of $C$ preceding them
are in $L$, and the three consecutive vertices of $C$ following
them are in $L$. That is, if $c_{i},c_{i+1}\notin L$, then
$c_{i-3},c_{i-2},c_{i-1},c_{i+2},c_{i+3},c_{i+4}\in L$.
\end{condition}

Note that in particular, the condition implies that there does not exist a set
of three or more consecutive vertices of $C$ such that none of them is in $L$.

\begin{lemma}
\label{l1k2ap}For $n\geq6$, if $L\in LS_{2}^{AP}(G)$, then Condition
\ref{APk2} holds.
\end{lemma}

\begin{proof}
We show that the condition is necessary for a landmark set $L$. Assume that
$c_{i},c_{i+1}\notin L$. We show that $c_{i+2},c_{i+3},c_{i+4}\in L$ (the
proof for the vertices $c_{i-3},c_{i-2},c_{i-1}$ is similar, and note that the
two sets are not necessarily disjoint, and if $n=6$, this is the same set). To
obtain two separations between $c_{i+1}$ and $c_{i+2}$, both $c_{i+2}$ and
$c_{i+3}$ must be landmarks, and to obtain two separations between $c_{i+1}$
and $c_{i+3}$, $c_{i+3}$ and $c_{i+4}$ must be landmarks (since $c_{i}$ and
$c_{i+1}$ are not landmarks).
\end{proof}

We analyze the cardinality of sets that satisfy Condition \ref{APk2}.

\begin{proposition}
\label{abck2ap} Any set satisfying Condition \ref{APk2} has at least
$\lfloor\frac{n}{2}\rfloor$ cycle vertices. If $L\subseteq V$ is a set where
$|L\cap C|\geq4$, then there are at least two separations in $L\cap C$ for any
pair $h,c_{i}$, and in this case, if $L$ is a landmark set, then
$L\setminus\{h\}$ is a landmark set.
\end{proposition}
\begin{proof}
We start with the first part. Consider a set $X$ that satisfies the condition.
If $C\subseteq X$, then we are done as $n-1\geq\frac{n}{2}$ for $n\geq5$. If
$X\cap C=\emptyset$, then since $|C|\geq5$, there are at least three
consecutive cycle vertices not in $X$, and the condition is not satisfied.
Otherwise, let $c_{i+1}$ be a cycle vertex such that $c_{i+1}\in X$ while
$c_{i}\notin X$. Create a binary string of length $n-1$, where the $j$th bit
of the string is $1$ if $c_{i+j}\in X$ and $0$ otherwise. The string starts
with $1$ and ends with $0$. Partition the string into maximum length
substrings that each substring starts with $1$, ends with $0$, and contains
ones followed by zeroes. A substring never has more than two zeroes, and if it
has two zeroes, then it has at least three ones preceding them. Thus, the
number of ones is at least the number of zeroes, proving that there are at
least $\lceil\frac{n-1}{2}\rceil=\lfloor\frac{n}{2}\rfloor$ ones.

To prove the second part, we note that $h$ can only separate a pair of
vertices of the form $h,c_{i}$. As $L\setminus\{h,c_{i-1},c_{i+1}\}$ contains
at least two cycle vertices, there are at least two separations between $h$
and $c_{i}$ even without $h$.
\end{proof}

\begin{lemma}
\label{l2k2ap} Let $n\geq6$. For a set $L\subseteq C$, if Condition \ref{APk2}
holds for $L$, then for any pair $c_{i},c_{j}$ there are two separations.
\end{lemma}

\begin{proof}
Consider a set $L$ satisfying the condition. We will show that any pair of
cycle vertices has at least two separations. Consider a pair $c_{i},c_{i+\ell
}$, where $\ell=1$ or $\ell=2$, and consider the list of vertices
$c_{i-1},c_{i},c_{i+\ell},c_{i+\ell+1}$. We show that $L$ has at least two
vertices out of this list. If $L$ contains at least one vertex out of the
first two vertices of the list, and at least one vertex of the last two
vertices of the list, then we are done. If the first two vertices of the list
are not in $L$, then by the condition, the vertices $c_{i+1}$, $c_{i+2}$,
$c_{i+3}$ must be in $L$, so the other two vertices of the list that can
separate $c_{i}$ and $c_{i}+\ell$ are in $L$, and similarly, if $c_{i+\ell
},c_{i+\ell+1}$ are not in $L$, then the other two vertices of the list are in
$L$. This shows that in these cases, at least two out of the four vertices are
in $L$. Consider a distant pair $c_{i},c_{j}$. By the condition,
$c_{i-1},c_{i},c_{i+1}$ must have at least one vertex in $L$, and
$c_{j-1},c_{j},c_{j+1}$ must have at least one vertex in $L$, giving at least
two separations between $c_{i}$ and $c_{j}$.
\end{proof}

Based on the above, we get the following.

\begin{corollary}
\label{n8ap} For $n\geq8$, $L\subseteq C$ is a landmark set if and
only if it satisfies Condition \ref{APk2}. Moreover,
$md_{2}^{AP}=\lfloor\frac{n}{2}\rfloor$.
\end{corollary}

\begin{proof}
A set that satisfies the condition has at least
$\lfloor\frac{n}{2}\rfloor \geq4$ cycle vertices. Such a set is a
landmark set; by Proposition \ref{abck2ap}, there are two
separations for any pair of the form $h,c_{i}$, and by Lemma
\ref{l2k2ap}, any pair of cycle vertices also has at least two
separations. By Lemma \ref{l1k2ap}, any landmark set satisfies the
condition. This proves the first part, and we find
$md_{2}^{AP}\geq\lfloor\frac{n}{2}\rfloor$. The set
$L=\{c_{i}|1\leq i\leq n-1,i\mbox{ \ is odd}\}$ satisfies the
condition (and thus it is a landmark set) and it has
$\lceil\frac{n-1}{2}\rceil=\lfloor\frac{n}{2}\rfloor$ vertices,
proving that for $n\geq8$,
$md_{2}^{AP}\leq\lfloor\frac{n}{2}\rfloor$.
\end{proof}

Next, we discuss minimal landmark sets for $n=6,7,8$. By corollary \ref{n8ap},
for $n=8$, such sets do not contain $h$, and the sets of cycle vertices that
satisfy the condition have at least four (cycle) vertices. Any subset of at
least five cycle vertices is a landmark set, but only those with two
consecutive non-landmarks are minimal, with respect to set inclusion (sets of
the form $c_{j},c_{j+1},\ldots,c_{j+4}$ for $1\leq j\leq7$), since if there
are no consecutive non-landmarks, there must be three consecutive landmarks
$c_{i-1},c_{i},c_{i+1}$, and $c_{i}$ can be omitted. Any subset of four cycle
vertices for which there is no pair of consecutive non-landmarks satisfies the
condition, and it has the form $c_{i},c_{i+1},c_{i+3},c_{i+5}$ for some $1\leq
i\leq7$.

For $n=6,7$, sets a set satisfying Condition \ref{APk2} has at least three
cycle vertices, and we consider such sets. The next Corollary specifies all
minimal landmark sets for $n=6,7$, where unlike the case $n\geq8$, there exist
such sets that contain $h$.

\begin{corollary}
For $n=6,7$, given $A\subseteq C$, if $|A|\geq4$, then $A$ is a landmark set.
Assume that $|X|=3$, and $X$ satisfies Condition \ref{APk2}. The set
$Y=X\cup\{h\}$ is a landmark set, while $X$ is not a landmark set.
\end{corollary}

\begin{proof}
If $n=6,7$, any set with four cycle vertices satisfies the condition; for
$n=6$, there is only one cycle vertex not in $A$, and for $n=7$, there are at
most two cycle vertices not in $A$, and the condition is satisfied no matter
whether they are adjacent or not. By Proposition \ref{abck2ap}, there are two
separations for any pair of the form $h,c_{i}$. By Lemma \ref{l2k2ap}, any
pair of cycle vertices also has at least two separations.

The set $X$ has two separations for any pair of cycle vertices, by Lemma
\ref{l2k2ap}, and in $Y$, a pair $h,c_{j}$ is separated by $h$ and
$X\setminus\{c_{j-1},c_{j+1}\}$, which must contain at least one vertex, so
$Y$ is a landmark set. For $n=7$, the only form of a set $X$ that satisfies
the condition is $c_{i},c_{i+2},c_{i+4}$, since in the case of two consecutive
vertices not in $X$, all remaining cycle vertices must be in $X$. For this
set, $h$ and $c_{i+1}$ are only separated by $c_{i+4}$. For $n=6$, $X$ can be
of one of two forms. If $X$ consists of $c_{i},c_{i+1},c_{i+3}$, then
$c_{i+2}$ and $h$ are only separated by $c_{i}$. If $X$ consists of
$c_{i},c_{i+1},c_{i+2}$, then $c_{i+1},h$ are only separated by $c_{i+1}$.
Thus, $X$ is never a landmark set.
\end{proof}

Let $n\geq9$. For a landmark set (that does not contain $h$), we define a
cyclic binary string of length $n-1$ where similar to the proof of Proposition
\ref{abck2ap}, a zero denotes a non-landmark and a $1$ denote a landmark. A
string of five bits is called \textit{invalid} (for AP and $k=2$) if it
contains at least four zeroes, or if it contains three zeroes and it is not
equal to $01010$. Other strings of five bits are called good.

\begin{lemma}
\label{apk2word} Let $n\geq9$. A cyclic binary string represents a landmark
set if and only if it does not contain any invalid five bit string as a substring.
\end{lemma}

\begin{proof}
We start with showing that if the substring corresponding to a set
$X$ has an invalid substring, then $X$ does not satisfy Condition
\ref{APk2}, and therefore it is not a landmark set. We show that
an invalid substring always has a pair of consecutive zeroes. This
holds for $00000$, and it holds for any string with four zeroes
(as it only has one $1$). The only string of length five and three
zeroes that does not have two consecutive zeroes is $01010$, which
is a good substring. Since any invalid substring has at least one
additional zero except for the pair of consecutive zeroes, this
means that there are two cycle vertices $c_{i},c_{i+1}\notin X$
(corresponding to the pair of consecutive zeroes), such that one
of $c_{i-3},c_{i-2},c_{i-1},c_{i+2},c_{i+3},c_{i+4}$ has a zero in
the substring and thus it is not in $X$, violating Condition
\ref{APk2}.

Assume now that a set $Y\subseteq C$ does not satisfy Condition
\ref{APk2}. So, there is a pair $c_{i},c_{i+1}\notin Y$ such that
one of $c_{i-3},c_{i-2},c_{i-1},c_{i+2},c_{i+3},c_{i+4}$ is not in
$Y$. If one of $c_{i-3},c_{i-2},c_{i-1}$ is not in $Y$, then the
substring for $c_{i-3},c_{i-2},c_{i-1},c_{i},c_{i+1}$ has at least
three zeroes, and the last two bits are zeroes. If one of
$c_{i+2},c_{i+3},c_{i+4}$ is not in $Y$, then the substring for
$c_{i},c_{i+1},c_{i+2},c_{i+3},c_{i+4}$ starts with two zeroes,
and it has at least one additional zero. In both cases we find an
invalid substring.
\end{proof}

In the next section we state a dynamic programming formulation for computing
$wmd_{2}^{AP}$ for a wheel $G$ with $n\geq9$ that is based on the cyclic
binary string, and on the fact that including $h$ in landmark sets is not
necessary. The remaining cases for AP ($k\geq3$) are simpler and we discuss
them now.

\subsubsection{The case $\boldsymbol{k=3}$}
Let $k=3$ and consider the following condition for $n\geq6$.

\begin{condition}
\label{k3ap} For any vertex of $C$ that is not in $L$, the four
consecutive vertices of $C$ preceding it, and the four consecutive
vertices of $C$ following it are in $L$. That is, if $c_{i}\notin
L$, then
$$c_{i-4},c_{i-3},c_{i-2},c_{i-1},c_{i+1},c_{i+2},c_{i+3},c_{i+4}\in
L.$$
\end{condition}

\begin{lemma}
\label{lemk3ap} \label{l1k3ap}For $n\geq6$, if $L\in LS_{3}^{AP}(G)$, then
Condition \ref{k3ap} holds.
\end{lemma}

\begin{proof}
We show that the condition is necessary for a landmark set $L$. Assume that
$c_{i}\notin L$. We show that $c_{i+1},c_{i+2},c_{i+3},c_{i+4}\in L$ (the
proof for $c_{i-4},c_{i-3},c_{i-2},c_{i-1}$ is similar). To obtain three
separations between $c_{i}$ and $c_{i+\ell}$ for $\ell=1$ or $\ell=2$, at
least three of $c_{i-1},c_{i},c_{i+\ell},c_{i+\ell+1}$ must be in $L$, and
thus at least two of $c_{i},c_{i+\ell},c_{i+\ell+1}$ must be in $L$, that is
(since $c_{i}\notin L$), $c_{i+\ell},c_{i+\ell+1}\in L$ for $\ell=1,2$,
proving $c_{i+1},c_{i+2},c_{i+3}\in L$. To obtain three separations between
$c_{i+1}$ and $c_{i+3}$, at least three of $c_{i},c_{i+1},c_{i+3},c_{i+4}$
must be in $L$, and since $c_{i}\notin L$, $c_{i+4}\in L$.
\end{proof}

We analyze the cardinality of sets that satisfy the condition.

\begin{proposition}
\label{abck3ap} Any set satisfying the Condition \ref{k3ap} has at least
$\lfloor\frac{4n}{5}\rfloor$ cycle vertices. If $L\subseteq V$ is a set where
$|L\cap C|\geq5$, then there are at least three separations in $L\cap C$ for a
pair $h,c_{i}$, and in this case, if $L$ is a landmark set, then
$L\setminus\{h\}$ is a landmark set.
\end{proposition}
\begin{proof}
We start with the first part. Consider a set $X$ that satisfies Condition
\ref{k3ap}. If $C\subseteq X$, then we are done as $n-1\geq\frac{4n}{5}$ for
$n\geq5$. If $X\cap C=\emptyset$, then since $|C|\geq5$, there are at least
two consecutive cycle vertices not in $X$, and the condition is not satisfied.
Otherwise, let $c_{i+1}$ be a cycle vertex such that $c_{i+1}\in X$ while
$c_{i}\notin X$. Create a binary string of length $n-1$ for $X$. The string
starts with four times $1$ and ends with $0$. Partition the string into
maximum length substrings that each substring starts with $1$, ends with $0$,
and contains ones followed by zeroes. Every sequence of ones appears after a
zero in this cyclic string, and therefore it contains at least four ones by
the condition, while a substring never has more than one zero. Thus, the
number of ones is at least four times the number of zeroes, proving that there
are at most $\lfloor\frac{n-1}{5}\rfloor$ zeroes (and we are done since
$n-1-\lfloor\frac{n-1}{5}\rfloor=\lfloor\frac{4n}{5}\rfloor$).

To prove the second part, as $L \setminus\{h,c_{i-1},c_{i+1}\}$ contains at
least three cycle vertices, there are at least three separations between $h$
and $c_{i}$ even without $h$.
\end{proof}

\begin{lemma}
\label{l2k3ap} Let $n \geq6$. For a set $L\subseteq C$, if Condition
\ref{k3ap} holds for $L$, then for any pair $c_{i},c_{j}$ there are at least
three separations.
\end{lemma}

\begin{proof}
Consider a set satisfying the condition. We will show that any
pair of cycle vertices has at least three separations. Consider a
pair $c_{i},c_{i+\ell}$, where $\ell=1$ or $\ell=2$, and consider
the vertices $c_{i-1},c_{i},c_{i+\ell},c_{i+\ell+1}$. By the
condition, at most one of them is not in $L$ (as a vertex not in
$L$ has at least four vertices in $L$ before it and following it),
giving three separations. Consider a distant pair $c_{i},c_{j}$.
By the condition, $c_{i-1},c_{i},c_{i+1}$ must have at least two
vertices in $L$, and $c_{j-1},c_{j},c_{j+1}$ must have at least
two vertices in $L$, giving at least four separations between
$c_{i}$ and $c_{j}$.
\end{proof}

\begin{corollary}
For $n\geq7$, $L\subseteq C$ is a landmark set if and only if it satisfies the
condition. Moreover, $md_{3}^{AP}=\lfloor\frac{4n}{5}\rfloor$.
\end{corollary}

\begin{proof}
If $n\geq7$, then a set that satisfies the condition has at least
$\lfloor\frac{4n}{5}\rfloor\geq5$ cycle vertices (using Proposition
\ref{abck3ap}), there are at least three separations for any pair of the form
$h,c_{i}$. By Lemma \ref{l2k3ap}, any pair of cycle vertices also has at least
two separations. Thus, a set satisfying the condition is a landmark set, and
by Lemma \ref{lemk3ap}, every landmark set satisfies the condition. This
proves the first part, and we find $md_{2}^{AP}\geq\lfloor\frac{4n}{5}\rfloor
$. The set $L=\{c_{i}|1\leq i\leq n-1,i\mbox{ \ is not divisible by }5\}$
satisfies the condition and has $\left\lceil \frac{4(n-1)}{5}\right\rceil
=\left\lfloor \frac{4n}{5}\right\rfloor $ vertices.
\end{proof}

For $n=7,8$, a minimal landmark set (with respect to set inclusion) consists
of $n-2$ cycle vertices, and any subset of $n-2$ vertices is a landmark set.
For $n=6$, a set that satisfies the condition has at least four cycle
vertices. We now specify the structure of minimal landmark sets for $n=6$, and
show these sets are exactly all subsets of five vertices.

\begin{corollary}
For $n=6$, a set of four cycle vertices is not a landmark set. Any subset of
$V$ with five vertices is a landmark set.
\end{corollary}

\begin{proof}
A set of four (consecutive) cycle vertices $c_{i-2},c_{i-1},c_{i+1},c_{i+2}$
only makes two separations for the following three pairs: $h,c_{i-2}$,
$h,c_{i}$, $h,c_{i+2}$. Any set of five vertices gives another separation to
each one of the pairs, as both $h$ and $c_{i}$ separate each one of the three pairs.
\end{proof}

A string of five bits is called \textit{invalid} (for AP and $k=3$) if it has
at least two zeroes.

\begin{lemma}
\label{apk3word} Let $n \geq9$. A cyclic binary string represents a landmark
set if and only if it does not contain any invalid five bit string as a substring.
\end{lemma}

\begin{proof}
If a substring has at least two zeroes, then each of the vertices
corresponding to the zeroes does not have at least four
consecutive ones before and after it. Condition \ref{k3ap} is not
satisfied, and therefore, the corresponding set is not a landmark
set. Assume now that a set $Y \subseteq C$ does not satisfy the
condition. There is a pair $c_{i},c_{i+\ell} \notin Y$ for some $1
\leq\ell\leq4$. The substring for
$c_{i},c_{i+1},c_{i+2},c_{i+3},c_{i+4}$ has two zeroes, and
therefore it is invalid.
\end{proof}

In the next section we state a dynamic programming formulation for computing
$wmd_{3}^{AP}$ based on the cyclic binary string, and on the fact that
including $h$ in landmark sets is not necessary for $n\geq9$.

\subsubsection{The case $\boldsymbol{k\geq4}$}

Finally, for $k\geq4$, we can prove the following.

\begin{theorem}
For AP, and $n\geq6$, $md_{k}^{AP}=\infty$ for $k\geq5$. Additionally,
$md_{4}^{AP}=n-1$ if $n\geq7$ (in this case $wmd_{4}^{AP}=w(C)$), and
$md_{4}^{AP}=6$ for $n=6$ (in this case $wmd_{4}^{AP}=w(V)$).
\end{theorem}

\begin{proof}
Since a pair of the form $c_{i},c_{i+1}$ can be separated by exactly four
vertices, a landmark set for $k=4$ must contain all cycle vertices, and there
is no valid landmark set for $k\geq5$. Any cycle vertex except for $c_{i+1}$
and $c_{i-1}$ separates $h$ and $c_{i}$, thus $C$ is a landmark set for $k=4$
if $n\geq7$. For $n=6$, there are only three separations by cycle vertices
between $h$ and $c_{i}$, and therefore $V$ is the only landmark set (as $h$
adds a separation).
\end{proof}

\subsection{Non-landmarks model (NL)}

\subsubsection{The case $\boldsymbol{k=2}$}

Consider the following condition for $n\geq6$.
\begin{condition}
\label{NLk2} For any pair of consecutive vertices of $C$ that are
not in $L$, the two consecutive vertices of $C$ preceding them,
and the two consecutive vertices of $C$ following them are in $L$.
That is, if $c_{i},c_{i+1}\notin L$, then
$c_{i-2},c_{i-1},c_{i+2},c_{i+3}\in L$.
\end{condition}

Similarly to AP, the condition implies that there does not exist a set of
three or more consecutive vertices of $C$ such that none of them is in $L$.

\begin{lemma}
\label{l1k2nl}For $n\geq6$, if $L\in LS_{2}^{NL}(G)$, then Condition
\ref{NLk2} holds.
\end{lemma}

\begin{proof}
We show that the condition is necessary for a landmark set $L$. Assume that
$c_{i},c_{i+1}\notin L$. We show that $c_{i+2},c_{i+3}\in L$ (the proof for
the vertices $c_{i-1},c_{i-2}$ is similar, and here the two sets are again not
necessarily disjoint). Assume by contradiction that $c_{i+2}\notin L$. Then,
only $c_{i+3}$ can separate between $c_{i+1}$ and $c_{i+2}$, a contradiction.
Assume by contradiction that $c_{i+3}\notin L$. Then, only $c_{i+4}$ can
separate between $c_{i+1}$ and $c_{i+3}$, a contradiction.
\end{proof}

We analyze the cardinality of sets that satisfy the condition.

\begin{proposition}
\label{abck2nl} Any set satisfying Condition \ref{NLk2} has at least
$\lfloor\frac{n}{2}\rfloor$ cycle vertices. If $L\subseteq V$ is a set where
$|L\cap C|\geq4$, then there are at least two separations in $L\cap C$ for a
pair $h,c_{i}$, and in this case, if $L$ is a landmark set, then
$L\setminus\{h\}$ is a landmark set.
\end{proposition}
\begin{proof}
We start with the first part. Consider a set $X$ that satisfies the condition.
If $C\subseteq X$, then we are done as $n-1\geq\frac{n}{2}$ for $n\geq5$. If
$X\cap C=\emptyset$, then since $|C|\geq5$, there are at least three
consecutive cycle vertices not in $X$, and the condition is not satisfied.
Otherwise, let $c_{i+1}$ be a cycle vertex such that $c_{i+1}\in X$ while
$c_{i}\notin X$. Once again we create a binary string of length $n-1$ that
starts with $1$ and ends with $0$, and partition the string into maximum
length substrings that each substring starts with $1$, ends with $0$, and
contains ones followed by zeroes. A substring never has more than two zeroes,
and if it has two zeroes, then it has at least two ones. Thus, the number of
ones is at least the number of zeroes, proving that there are at least
$\lceil\frac{n-1}{2}\rceil=\lfloor\frac{n}{2}\rfloor$ ones.

To prove the second part, as $h$ cannot separate cycle vertices, and as
$L\setminus\{h,c_{i-1},c_{i+1}\}$ contains at least two cycle vertices, there
are at least two separations between $h$ and $c_{i}$ even without $h$.
\end{proof}

\begin{lemma}
\label{l2k2nl} Let $n\geq6$. For a set $L\subseteq C$, if Condition \ref{NLk2}
holds, there are at least two separations for any pair $c_{i},c_{j}\notin L$.
\end{lemma}

\begin{proof}
Consider a set satisfying the condition. We show that any pair of
cycle vertices not in $L$ has at least two separations. Consider a
pair $c_{i},c_{i+\ell}\notin L$, where $\ell=1$ or $\ell=2$. If
$c_{i-1}\notin L$, then by the condition for the two vertices
$c_{i-1},c_{i}$, we have $c_{i+1},c_{i+2}\in L$, contradicting
$c_{i+\ell}\notin L$. Similarly, if $c_{i+\ell+1}\notin L$, then
by the condition for the two vertices $c_{i+\ell },c_{i+\ell+1}$,
we have $c_{i+\ell-1},c_{i+\ell-2}\in L$, a contradiction as one
of the last two vertices is the vertex $c_{i}$. We find
$c_{i-1},c_{i+\ell+1}\in L$, and they are distinct vertices
separating $c_{i}$ and $c_{i+\ell}$. Next, consider a distant pair
$c_{i},c_{j}\notin L$. By the condition, $c_{i-1},c_{i},c_{i+1}$
must have at least one vertex in $L$, and $c_{j-1},c_{j},c_{j+1}$
must have at least one vertex in $L$, giving at least two
separations between $c_{i}$ and $c_{j}$.
\end{proof}

Based on the above, we get the following.

\begin{corollary}
For $n\geq8$, $L\subseteq C$ is a landmark set if and only if it satisfies
Condition \ref{NLk2}. Moreover, $md_{2}^{NL}=\lfloor\frac{n}{2}\rfloor$.
\end{corollary}

\begin{proof}
A set that satisfies the condition has at least
$\lfloor\frac{n}{2}\rfloor \geq4$ cycle vertices. Such a set is a
landmark set as by Proposition \ref{abck2nl}, there are two
separations for any pair of the form $h,c_{i}$, and by Lemma
\ref{l2k2nl}, any pair of cycle vertices also has at least two
separations. By Lemma \ref{l1k2nl}, any landmark set satisfies the
condition. This proves the first part, and we find
$md_{2}^{NL}\geq\lfloor\frac{n}{2}\rfloor$, while the other
inequality follows as there is a landmark set of this cardinality
for AP.
\end{proof}

The minimal landmark sets (with respect to set inclusion) for $n=8$ are not
the same as for AP, since the condition is slightly weaker. All subsets of
four cycle vertices where there are no pair of consecutive non-landmark (that
were defined for AP) are still minimal landmark sets. Minimal landmark sets
with a pair of consecutive non-landmarks are possible too, and have the form
$c_{i},c_{i+1},c_{i+3},c_{i+4}$ (for some $1\leq i\leq7$).

For $n=6,7$, we consider sets with at least three cycle vertices.

\begin{corollary}
For $n=6,7$, given $A\subseteq C$, if $|A|\geq4$, then $A$ is a landmark set.
Assume that $|X|=3$, and $X$ satisfies Condition \ref{NLk2}. The set
$Y=X\cup\{h\}$ is a landmark set, while $X$ is a landmark set only if $n=6$,
and it consists of three consecutive cycle vertices.
\end{corollary}

\begin{proof}
The subsets of four vertices are landmark sets as they were proved to be
landmark sets for AP.

For $n=7$, the proof that any landmark set has at least four
vertices is the same as for AP. For $n=6$, the proof for a triple
of cycle vertices that are not consecutive is the same as for AP.
If $X$ consists of $c_{i},c_{i+1},c_{i+2}$, then $c_{i-1},h$ are
separated by $c_{i+1}$ and $c_{i+2}$, while $c_{i+3},h$ are
separated by $c_{i}$ and $c_{i+1}$, and $c_{i-1},c_{i+3}$ are
separated by $c_{i}$ and $c_{i+2}$.
\end{proof}

We find that a minimal landmark set (with respect to set inclusion) for $n=7$
is the same as for AP. For $n=6$, a minimal landmark set consists of either
three consecutive cycle vertices, or $h$ together with three cycle vertices
that are not consecutive.

A string of four bits is called \textit{invalid} (for NL and $k=2$) if it
contains at least three zeroes.

\begin{lemma}
\label{nlk2word} Let $n\geq9$. A cyclic binary string represents a landmark
set if and only if it does not contain any invalid four bit string as a substring.
\end{lemma}

\begin{proof}
Let $\alpha$ be an invalid substring. In this case, $\alpha$ has a pair of
consecutive zeroes, and at least one additional zero. This means that there
are two cycle vertices $c_{i},c_{i+1}\notin X$, such that one of
$c_{i-2},c_{i-1},c_{i+2},c_{i+3}$ has a zero corresponding to it in $\alpha$,
thus it is not in $X$, violating Condition \ref{NLk2}.

Assume now that a set $Y\subseteq C$ does not satisfy the
condition. Thus, there is a pair $c_{i},c_{i+1}\notin Y$ such that
one of $c_{i-2},c_{i-1},c_{i+2},c_{i+3}$ is not in $Y$. If one of
$c_{i-2},c_{i-1}$ is not in $Y$, then the substring for
$c_{i-2},c_{i-1},c_{i},c_{i+1}$ has at least three zeroes. If one
of $c_{i+2},c_{i+3}$ is not in $Y$, then the substring for
$c_{i},c_{i+1},c_{i+2},c_{i+3}$ has three zeroes. In both cases we
find an invalid substring.
\end{proof}

In the next section we state a dynamic programming formulation for computing
$wmd_{2}^{NL}$ based on the cyclic binary string, and on the fact that
including $h$ in any landmark set is not necessary for $n\geq9$.

\subsubsection{The cases $\boldsymbol{k=3,4}$}

Let $k\in\{3,4\}$ and consider the following condition.

\begin{condition}
\label{k34nl} For any vertex of $C$ that is not in $L$, the two
consecutive vertices of $C$ preceding it, and the two consecutive
vertices of $C$ following it are in $L$. That is, if $c_{i}\notin
L$, then $c_{i-2},c_{i-1},c_{i+1},c_{i+2}\in L$.
\end{condition}

\begin{lemma}
\label{l1k3nl}For $n\geq5$, if $L\in LS_{k}^{NL}(G)$ for $k=3$ or $k=4$, then
Condition \ref{k34nl} holds.
\end{lemma}

\begin{proof}
We show that the condition is necessary for a landmark set $L$. Assume that
$c_{i}\notin L$. We show that $c_{i+1},c_{i+2}\in L$ (the proof for
$c_{i-2},c_{i-1}$ is similar). If $c_{i+\ell}\notin L$ for $i\in\{1,2\}$, then
there can be at most two separations between $c_{i}$ and $c_{i+\ell}$ (those
are $c_{i-1}$ and $c_{i+\ell+1}$). Thus, as there must be at least three
separations for every pair of non-landmarks and $c_{i}\notin L$, we find that
$c_{i+1},c_{i+2}\in L$.
\end{proof}

We analyze the cardinality of sets that satisfy the condition.

\begin{proposition}
\label{abck3nl} Any set satisfying Condition \ref{k34nl} has at least
$\lfloor\frac{2n}{3}\rfloor$ cycle vertices. If $L\subseteq V$ is a set where
$|L\cap C|\geq k+2$, then there are at least $k$ separations in $L\cap C$ for
a pair $h,c_{i}$, and in this case, if $L$ is a landmark set, then
$L\setminus\{h\}$ is a landmark set.
\end{proposition}
\begin{proof}
We start with the first part. Consider a set $X$ that satisfies the condition.
If $C \subseteq X$, then we are done as $n-1 \geq\frac{2n}3$ for $n \geq5$. If
$X \cap C =\emptyset$, then since $|C| \geq5$, there are at least two
consecutive cycle vertices not in $X$, and the condition is not satisfied.
Otherwise, let $c_{i+1}$ be a cycle vertex such that $c_{i+1} \in X$ while
$c_{i} \notin X$. Create a binary string of length $n-1$ for $X$. The string
starts with $1$ and ends with $0$. Partition the string into maximum length
substrings that each substring starts with $1$, ends with $0$, and contains
ones followed by zeroes. A substring never has more than one zero. Thus, the
number of ones is at least twice the number of zeroes, proving that there are
at most $\lfloor\frac{n-1}3 \rfloor=n-1-\lfloor\frac{2n}3 \rfloor$ zeroes.

To prove the second part, as $L \setminus\{h,c_{i-1},c_{i+1}\}$ contains at
least $k$ cycle vertices, there are at least $k$ separations between $h$ and
$c_{i}$ even without $h$.
\end{proof}

\begin{lemma}
\label{l2k3nl} Let $n \geq6$. For a set $L\subseteq C$, if Condition
\ref{k34nl} holds for $L$, then for any pair $c_{i},c_{j}$ there are at least
four separations.
\end{lemma}

\begin{proof}
Consider a set satisfying the condition. We will show that any pair of
non-landmark cycle vertices has at least four separations. Consider a pair
$c_{i},c_{i+\ell}$, where $\ell=1$ or $\ell=2$. By the condition, it is
impossible that both are non-landmarks. Thus, we are left with the case of a
distant pair $c_{i},c_{j}$. By the condition, $c_{i-1},c_{i},c_{i+1}$ must
have at least two vertices in $L$, and $c_{j-1},c_{j},c_{j+1}$ must have at
least two vertices in $L$, giving at least four separations between $c_{i}$
and $c_{j}$.
\end{proof}

\begin{corollary}
For $n\geq9$, $L\subseteq C$ is a landmark set if and only if it satisfies the
Condition \ref{k34nl}. Moreover, $md_{k}^{NL}=\lfloor\frac{2n}{3}\rfloor$  for
$k=3,4$. For $n=8$ and $k=3$, $L\subseteq C$ is a landmark set if and only if
it satisfies the condition. In addition, $md_{3}^{NL}=5$.
\end{corollary}

\begin{proof}
Let $n\geq9$. A set that satisfies the condition has at least $\lfloor
\frac{2n}{3}\rfloor\geq6$ cycle vertices. Such a set is a landmark set as by
Proposition \ref{abck3nl}, there are at least four separations for any pair of
the form $h,c_{i}$, and by Lemma \ref{l2k3nl}, any pair of non-landmark cycle
vertices also has at least two separations. By Lemma \ref{l1k3nl}, any
landmark set satisfies the condition. This proves the first part, and we find
$md_{k}^{NL}\geq\lfloor\frac{2n}{3}\rfloor$, for $k=3,4$. The set
$L=\{c_{i}|1\leq i\leq n-1,$ $i\mbox{ \ is not divisible
by }3\}$ satisfies the condition and has $\lceil\frac{2(n-1)}{3}\rceil
=\lfloor\frac{2n}{3}\rfloor$ vertices.

If $n=8$ and $k=3$, a set that satisfies the condition has at least five cycle
vertices (using Proposition \ref{abck3nl}), and there are three separations
for any pair of the form $h,c_{i}$. By Lemma \ref{l2k3nl}, any pair of
non-landmark cycle vertices also has three separations. The set $L$ above
(which is $\{c_{1},c_{2},c_{4},c_{5},c_{7}\}$) is a landmark set with five vertices.
\end{proof}

It is left to consider the cases $n=6,7,8$, $k=4$, and $n=6,7$, $k=3$.

\begin{lemma}
For $n = 6$ and $k=3,4$, only trivial landmark sets exist ($md_{4}^{NL}=5$).

For $n=7$ and $k=3$, $L\subseteq V$ is a landmark set if and only if it
satisfies Condition \ref{k34nl} and $|L|\geq5$ (thus $md_{3}^{NL}=5$). For
$n=7$ and $k=4$, $L\subseteq V$ is a landmark set if and only if it either
satisfies the Condition \ref{k34nl} and in addition $h\in L$, or $L=C$ (thus
$md_{4}^{NL}=5$).

For $n=8$ and $k=4$, $L\subseteq V$ is a landmark set if and only if it either
satisfies Condition \ref{k34nl} and in addition $h\in L$, or it has at least
six cycle vertices (thus $md_{4}^{NL}=6$).
\end{lemma}

\begin{proof}
For $n=6$, a set $L$ that satisfies the condition has at least four cycle
vertices. If $c_{i},h \notin L$, then as $c_{i-1},c_{i+1}$ do not separate
$c_{i}$ and $h$, there are two separations between them. Thus, $L$ must
contain another vertex, $|L|\geq5$, and it is a trivial landmark set.

For $n=7$, a set $L$ that satisfies the condition has at least four cycle
vertices. If it has exactly four cycle vertices, then it has the form
$c_{i-2},c_{i-1},c_{i+1},c_{i+2}$. If $h \notin L$, then similarly to the case
$n=6$, there are two separations between $c_{i}$ and $h$. On the other hand if
$h \in L$, this is a landmark set for $k=3,4$. Assume now that $h \notin L$,
and thus $L$ must contain at least five cycle vertices. If $c_{i} \notin L$,
then there are three separations between $h$ and $c_{i}$, so $L$ is a landmark
set for $k=3$ but not for $k=4$, and the only landmark set for $k=4$ not
containing $h$ is $C$.

For $n=8$ and $k=4$, a set $L$ that satisfies the condition has at least five
cycle vertices. If $h \in L$, then it is a landmark set. Otherwise, by the
condition, if $c_{i} \notin L$, then $c_{i-1},c_{i+1} \in L$, and to obtain
four separations between $c_{i}$ and $h$, $L$ must contain $C\setminus
\{c_{i}\}$, which results in a landmark set of this form.
\end{proof}

A string of three bits is called \textit{invalid} string (for NL and $k=3,4$)
if it has at least two zeroes.

\begin{lemma}
\label{nlk34word} Let $n \geq9$. A cyclic binary string represents a landmark
set if and only if it does not contain any invalid three bit string as a substring.
\end{lemma}

\begin{proof}
If a substring has at least two zeroes, then each of the vertices
corresponding to them is a non-landmark that does not have at least two
consecutive ones before and after it, thus the corresponding set is not a
landmark set. Assume now that a set $Y\subseteq C$ does not satisfy the
condition. There is a pair $c_{i},c_{i+\ell}\notin Y$ for some $1\leq\ell
\leq2$. The substring for $c_{i},c_{i+1},c_{i+2}$ has two zeroes, and
therefore it is invalid.
\end{proof}

In the next section we state a dynamic programming formulation for computing
$wmd_{3}^{AP}$ based on the cyclic binary string, and on the fact that
including $h$ in any landmark set is not necessary.

In this case the dynamic programming will be sufficiently simple, and we will
not use cyclic binary strings.

\subsubsection{The case $\boldsymbol{k\geq5}$}
For $k\geq5$, we can prove the following.
\begin{theorem}
For NL and $n\geq6$, $md_{k}^{NL}=n-2$ if $n\geq k+4$, and $md_{k}^{NL}=n-1$
if $n\leq k+3$.
\end{theorem}

\begin{proof}
Since a pair of cycle vertices $c_{i},c_{j}$ that are not
landmarks can be separated by exactly four vertices, any landmark
set must contain all cycle vertices except for at most one. Any
set consisting of $C\setminus\{c_{i}\}$ for some $1\leq i\leq n-1$
is a landmark set if and only if $h$ and $c_{i}$ can be separated
by $k$ cycle vertices. The number of cycle vertices that can
separate them is $n-4$ (all cycle vertices except for
$c_{i-1},c_{i},c_{i+1}$). Thus, if $k\geq n-3$, only trivial
landmark sets exist, and otherwise the minimal landmark sets
consist of all cycle vertices expect for one.
\end{proof}

An algorithm for the weighted version acts as follows. If $n\leq8$, then it
tests all minimal landmark sets defined above, and returns a landmark set of
minimum weight. For each $n=4,5,6,7,8$, these minimal sets are described
above. For $n\geq9$, we showed that any minimal landmark set (with respect to
set inclusion) does not contain $h$. Thus, we test all subsets of cycle
vertices according to the relevant conditions.

\subsection{Dynamic programming formulations for $\boldsymbol{n
\geq9}$}

The general structure of the dynamic programming formulations is
as follows. For a specific variant, consider the list of good
strings for it of length $q$ (where the list of invalid strings of
length $q$ is its complement set). Let $S$ be the set of strings
of $q-1$ bits, each of which being a substring of a good string of
$q$ bits (each substring consists of the first $q-1$ bits or the
last $q-1$ bits of at least one good $q$ bit string). Recall that
for all cases $q\leq5$, and therefore $|S|\leq16$. We will compute
$|S|^{2}$ functions $F_{\bar{\gamma}}^{\bar{\beta}}(i)$, for
$\bar{\beta},\bar{\gamma}\in S$, and integral $i$ such that
$q-1\leq i\leq n$, where $F_{\bar{\gamma}}^{\bar{\beta}}(i)$
denotes the minimum weight of a valid string of length $i$
corresponding to $c_{1},c_{2},\ldots,c_{i}$ that does not have an
invalid substring (seeing it as a linear string, not as a cyclic
string), that starts with the string $\bar{\beta}$ and ends with
the string $\bar{\gamma}$. Obviously, as we deal with a cycle (of
the wheel) rather than a path, the string should be cyclic, and
the substring of the last $q-1$ bits (corresponding to
$c_{n-q+1},c_{n-q+2},\ldots,c_{n-1}$) concatenated with the first
$q-1$ bits (corresponding to $c_{1},c_{2},\ldots,c_{q-1}$) cannot
contain an invalid substring. For every pair
$\bar{\beta},\bar{\gamma}\in S$, we determine whether a cyclic
string starting with $\bar{\beta}$ and ending with $\bar{\gamma}$
would contain an invalid string. Let
${\hat{\mathcal{S}}}=\{(\bar{\beta},\bar{\gamma})|\bar{\beta},\bar{\gamma}\in
S,\ \bar{\gamma}\circ\bar{\beta}\mbox{ does not contain an invalid
string}\}$ (where $\circ$ denotes concatenation). Since $q-1\leq4$
and we deal with $n-1\geq8$, the first $q-1$ bits and the last
$q-1$ bits together correspond to eight distinct vertices. We let
$F_{\bar{\gamma}}^{\bar{\beta}}(i)=\infty$ for all $i$ if
$\bar{\gamma}\notin S$ or $\bar{\beta}\notin S$. Additionally, we
let $F_{\bar{\gamma}}^{\bar{\beta}}(q-1)=\infty$ if
$\bar{\beta}\neq\bar{\gamma}$, and for $\bar{\beta}\in S$, such
that $\bar{\beta}=\beta_{1},\beta_{2},\ldots,\beta_{q-1}$,
$F_{\bar{\beta}}^{\bar{\beta}}(q-1)=\sum_{i=1}^{q-1}\beta_{i}\cdot
w(c_{i})$ (this is the cost incurred by placing landmarks in the
vertices whose bits are equal to $1$).

The values of $F_{\bar{\gamma}}^{\bar{\beta}}(i)$ for $q\leq i\leq
n-1$, $\bar{\gamma},\bar{\beta}\in S$ where
$\bar{\gamma}=\gamma_{1},\gamma_{2},\ldots,\gamma_{q-1}$ are
defined as follows. $F_{\bar{\gamma}}^{\bar{\beta}}(i)=
\gamma_{q-1}\cdot w(c_{i})+\min \{F_{\xi}^{\bar{\beta}}(i-1),
F_{\delta}^{\bar{\beta}}(i-1)\}$, where $\xi$ consists of the
sequence $0,\gamma_{1},\gamma_{2},\ldots,\gamma_{q-2}$ and
$\delta$ consists of the sequence
$1,\gamma_{1},\gamma_{2},\ldots,\gamma_{q-2}$. The output is a
subset of vertices that is found via backtracking using the
minimum term out of $F_{\bar{\gamma}}^{\bar{\beta}}(n-1)$ for all
$\bar{\gamma},\bar{\beta}$ such that
$(\bar{\beta},\bar{\gamma})\in{\hat{\mathcal{S}}}$. The running
time is $O(n)$ for all cases, as the number of functions is
constant (in all cases in fact $|S|\leq11$, and thus there are at
most $121$ functions).

Next, we state the sets $S$ and $\hat{\mathcal{S}}$ for the four
algorithms. In the case $k=3,4$ and NL, by Lemma \ref{nlk34word},
good strings are three bit strings with at most one zero. Thus
$S=\{11,10,01\}$, and
${\hat{\mathcal{S}}}=\{(11,11),(11,10),(11,01),(10,11),(10,01),(01,11)\}$.
In the case $k=3$ and AP, by Lemma \ref{apk3word}, good strings
are five bit strings with at most one zero. Thus
$S=\{1111,1110,1101,1011,0111\}$, and
\[
{\hat{\mathcal{S}}}=\{(1111,\bar{\gamma})|\bar{\gamma} \in S\}
\cap\{((1110,1111), (1110,1101), (1110,1011),
\]
\[
(1110,0111),(1101,1111),(1101,1011),(1101,0111),(1011,1111),(1011,0111),(0111,1111)\}
\ .
\]
In the case $k=2$ and NL, by Lemma \ref{nlk2word}, good strings are four bit
strings with at most two zeroes. Thus $S=\{111,110,101,100,011,010,001\}$,
and
\[
{\hat{\mathcal{S}}}=\{(\bar{\beta},\bar{\gamma})|\bar{\beta}
\in\{110,111\}, \bar{\gamma} \in S\} \cap\{(101,\bar{\gamma})|
\bar{\gamma} \in S \setminus\{100\}\}\cap\{ (100,\bar{\gamma}|
\bar{\gamma} \in S \setminus\{100, 110, 010\} \}
\]
\[
\cap\{(011,011),(011,101),(011,110),(011,111), (010,101),(010,111),(010,011),
(001,011),(001,111) \} \ .
\]
In the case $k=2$ and AP, by Lemma \ref{apk2word}, good strings are five bit
strings with at most two zeroes, or the string $01010$. Thus $S$ consists of
all four bit strings with at most two zeroes ($11$ strings, that is, the four
bit strings that are not in $S$ are $\{0000, 0001, 0010, 0100, 1000\}$), and
\[
{\hat{\hat{\mathcal{S}}}}=\{ (\bar{\beta},\bar{\gamma}|
\bar{\beta} \in\{0101, 0110\}, \bar{\gamma} \in\{0101, 0111, 1011,
1101, 1111\} \}\cap\{(\bar{\beta},\bar{\gamma})|\bar{\beta}
\in\{1111,1110\}, \bar{\gamma} \in S\}
\]
\[
\cap\{(1101,\bar{\gamma})| \bar{\gamma} \in S
\setminus\{1100\}\}\cap\{ (\bar{\beta},\bar {\gamma}| \bar{\beta}
\in\{1010, 1011\}, \bar{\gamma} \in S \setminus\{1001, 1100\} \}
\]
\[
\cap\{(0111,\bar{\gamma})| \bar{\gamma} \in\{0101, 0111, 1011,
1101, 1110, 1111\}\} \cap\{(1001,\bar{\gamma})| \bar{\gamma}
\in\{0011,0111,1011,1111)\}
\]
\[
\cap\{(1100,\bar{\gamma})| \bar{\gamma} \in S
\setminus\{0110,1010,1100,1110\}\} \cap\{(0011,\bar{\gamma})|
\bar{\gamma} \in\{0111,1111\}\} \ .
\]

\end{document}